\documentclass[letter, 11pt]{article}
\usepackage{amsmath} 
\usepackage{amssymb}  
\usepackage{amsthm}
\usepackage[dvips]{graphicx}
\usepackage{times}
\usepackage{multirow}
\newcommand{\bhline}[1]{\noalign{\hrule height #1}}  
 
\newcommand{\Order}{\mathrm{O}}

\newcommand{\defeq}{\stackrel{\mbox{\scriptsize{\normalfont\rmfamily def. }}}{=}}

\newcommand{\shortqed}{\hfill \mbox{$\blacksquare$} \smallskip}

\newcommand{\ats}{\alpha t^*} 
\newcommand{\MDG}{\mathcal{G}} 
\newcommand{\ME}{\mathcal{E}} 
 
\newcommand{\Ni}{{\cal N}^-}
\newcommand{\No}{{\cal N}^+}
\newcommand{\deltai}{\delta^-}
\newcommand{\deltao}{\delta^+}
\newcommand{\dtv}{{\cal D}_{\rm tv}}

\newcommand{\GCD}{{\rm GCD}}

\newcommand{\Dim}{d}

\newtheorem{theorem}{Theorem}[section]
\newtheorem{lemma}[theorem]{Lemma}

\newtheorem{proposition}[theorem]{Proposition}
\newtheorem{observation}[theorem]{Observation}

\title{Total Variation Discrepancy of Deterministic Random Walks \\ for Ergodic Markov Chains}
\author{
 Takeharu Shiraga\footnote{
   Graduate School of Information Science and Electrical Engineering, 
   Kyushu University, Fukuoka, Japan\protect\\ 
  {\ttfamily \{takeharu.shiraga,yamauchi,kijima,mak\}@inf.kyushu-u.ac.jp}} \and 
 Yukiko Yamauchi\footnotemark[1] \and 
 Shuji Kijima\footnotemark[1] \and
 Masafumi Yamashita\footnotemark[1]
}
\begin{document}
\makeatletter
\makeatother
\maketitle
\begin{abstract}
Motivated by a derandomization of  Markov chain Monte Carlo (MCMC), this paper investigates {\em deterministic random walks}, which is a deterministic process analogous to a random walk. 
While there are several progresses on the analysis of the vertex-wise discrepancy (i.e., $L_\infty$ discrepancy),  
little is known about the {\em total variation discrepancy} (i.e., $L_1$ discrepancy), which plays a significant role in the analysis of an FPRAS based on MCMC.  
This paper investigates upper bounds of the $L_1$ discrepancy between the expected number of tokens in a Markov chain and the number of tokens in its corresponding deterministic random walk. 
First, we give a simple but nontrivial upper bound $\Order (mt^*)$ of the $L_1$ discrepancy for any ergodic Markov chains, where $m$ is the number of edges of the transition diagram and $t^*$ is the mixing time of the Markov chain. 
Then, we give a better upper bound $\Order(m\sqrt{t^*\log t^*})$ for {\em non-oblivious} deterministic random walks, if the corresponding Markov chain is ergodic and lazy.~We also present some lower bounds. 

\end{abstract}
%
%
\section{Introduction}
\paragraph{Background}
 Markov chain Monte Carlo (MCMC) is
   a powerful technique of designing randomized approximation algorithms
   for {\#P}-hard problems.
 Jerrum et al.~\cite{JVV86} showed
  the equivalence in the sense of the polynomial time computation
  between {\em almost} uniform generation and randomized approximate counting
  for self-reducible problems.
 A number of fully polynomial-time randomized approximation schemes
(FPRAS) based on their technique
  have been developed for {\#}P-hard problems,
  such as the volume of a convex body~\cite{DFK1991,LV06,CV15},
integral of a log-concave function~\cite{LV06},
  partition function of the Ising model~\cite{JS93}, and counting
bipartite matchings~\cite{JS96}.
 When designing an FPRAS based on the technique, it is important that
  the {\em total variation distance}
   of the approximate distribution from the target distribution
  is sufficiently small, and hence
   analyses of the mixing times of Markov chains are central issues
in a series of works on MCMC for FPRAS
   to guarantee a small total variation distance is small.
 See also Section~\ref{subsec:RWMC} for the terminology of Markov chains.

 In contrast,
   not many results are known about
  {\em deterministic} approximation algorithms for {\#}P-hard problems.
 A remarkable progress is the correlation decay technique,
   independently devised by Weitz~\cite{Weitz2006} and Bandyopadhyay
and Gamarnik~\cite{BG08}, and
  there are several recent developments on the technique.
 For counting $0$-$1$ knapsack solutions,
  Gopalan et al.~\cite{GKM10}, and
  Stefankovic et al.~\cite{SVV2012} gave
   deterministic approximation algorithms (see also~\cite{GKMSVV2011}).
 Ando and Kijima~\cite{AK14} gave an FPTAS based on approximate convolutions
  for computing the volume of a $0$-$1$ knapsack polytope.
 A direct derandomization of MCMC algorithms is not known yet, but
  it holds a potential for a general scheme of designing deterministic
approximation algorithms for {\#}P-hard problems.
 {\em Deterministic random
walks}~\cite{CS06,CDST07,DF09,CDFS10,KKM12,KKM13,SYKY13}
   may be used as a substitute for Markov chains, for the purpose.

\paragraph{Deterministic random walk}
 Deterministic random walk is a deterministic process analogous
  to a (multiple) random walk\footnote{``multiple random walk'' means independent random walks of many tokens.}.
 A configuration $\chi^{(t)} \in \mathbb{Z}_{\geq 0}^V$ of $M$ tokens
distributed over a (finite) vertex set $V$
   is deterministically updated from time $t$ to $t+1$ by routers
equipped on vertices.
 The router on a vertex $u \in V$
   deterministically serves tokens on $u$ to neighboring vertex $v$
   with a ratio (about) $P_{uv} \in [0,1]$ such that $\sum_{v \in V} P_{uv}=1$,
  i.e., $P=(P_{uv}) \in \mathbb{R}^{V \times V}$ is a transition
matrix (when $V$ is finite).
 See Section~\ref{subsec:detRW}
  for the detailed description of the model with which this paper is concerned.
 Note that the expected configuration $\mu^{(t)} \in \mathbb{R}_{\geq
0}^V$ of $M$ tokens in
  a multiple random walk at time $t$ is given by $\mu^{(t)}=\chi^{(0)}P^t$ on the
assumption that $\chi^{(0)}=\mu^{(0)}$.

 Cooper and Spencer~\cite{CS06}
  investigated the rotor-router model,
   which is a deterministic random walk corresponding to a simple random walk,
  and
  showed for the $d$-dimensional (infinite) integer lattice that the
maximum vertex-wise discrepancy
  $\|\chi^{(t)}-\mu^{(t)}\|_{\infty}$ is upper bounded by a constant $c_d$,
  which depends only on~$d$ but is independent of the total number of tokens.
 Later,
  it is shown that $c_1\simeq 2.29$ \cite{CDST07} and
  $c_2$ is about $7.29$ or $7.83$ depending on the routers~\cite{DF09}.
%
 On the other hand, Cooper et al.~\cite{CDFS10} gave an example of
  a rotor-router on the infinite $k$-regular tree, such that
  its vertex-wise discrepancy gets ${\rm \Omega}(\sqrt{kt})$ for an arbitrarily fixed~$t$.

 Motivated by general transition matrices,
 Kijima et al.~\cite{KKM12} investigated a rotor-router model on
finite multidigraphs,
  and gave a bound $\Order(n|{\cal A}|)$ of the vertex-wise discrepancy
  when $P$ is rational, ergodic and reversible,
  where $n=|V|$ and ${\cal A}$ denotes the set of multiple edges.
 For an arbitrary rational transition matrix $P$, 
 Kajino et al.~\cite{KKM13} gave an upper bound
  using the second largest eigenvalue $\lambda^*$ of $P$ and some
other parameters of~$P$.
 To deal with irrational transition probabilities,
  Shiraga et al.~\cite{SYKY13} presented a generalized notion of the
rotor-router model,
   which  they call {\em functional router model}.
 They gave a bound $\Order((\pi_{\max}/\pi_{\min})t^* \Delta)$ of the vertex-wise discrepancy 
   for a specific functional router model (namely, {\em SRT-router model})
  when $P$ is ergodic and reversible,
  where $t^*$ denotes the mixing rate of $P$ and
  $\pi_{\max}$ (resp. $\pi_{\min}$) is the maximum (resp. minimum)
element of the stationary distribution vector $\pi$ of $P$.
Using~\cite{SYKY13}, Shiraga et al.~\cite{SYKY14} discussed the time complexity of a simulation, in which they are concerned with an {\em oblivious} version, meaning that the states of routers are reset in each step
  while the deterministic random walk above mentioned carries over the
states of routers to the next step. 

 Similar, or essentially the same concepts have been independently developed
  in several literature, such as load-balancing, information spreading
and self-organization.
 Rabani et al.~\cite{RSW98} investigated the diffusive model for load balancing,
  which is an oblivious version of deterministic random walk, and showed for the model that the vertex-wise discrepancy is
$\Order\left(\Delta\log(n)/(1-\lambda^*)\right)$
  when  $P$ is symmetric and ergodic,
 where $\Delta$ is the maximum degree of the transition diagram of $P$.
 Friedrich et al.~\cite{FGS12} proposed the BED algorithm for load balancing,
  which uses some extra information in the previous time, and
 they gave $\Order(\Dim^{1.5})$ for hypercube and $\Order(1)$ for
constant dimensional tori.
 Akbari et al.~\cite{AB13} discussed the relation between the BED
algorithm and the rotor-router model, and
 gave the same bounds for a rotor-router model.
 Berenbrink et al.~\cite{BKKM15} investigated
   about cumulatively fair balancers algorithms, which includes the
rotor-router model, and
   gave an upper bound $\Order(d\min(\sqrt{\log(n)/(1-\lambda^*)}, \sqrt{n}))$
  for a lazy version of simple random walks on $d$-regular graphs.

 As a closely related topic,
  the behavior of the rotor-router model with a single token has also
been investigated.
 Holroyd and Propp~\cite{HP10}
  investigated the frequency $\nu^{(t)} \in \mathbb{Z}_{\geq 0}^V$ of
visits of the token in $t$ steps, and
  showed that $\|\nu^{(t)}/t-\pi\|_{\infty}$ is $\Order(mn/t)$.
 Preceding~\cite{HP10},
  Yanovski et al.~\cite{YWB03} showed that
  the rotor-router model with a single token
   always stabilizes to a traversal of an Eulerian cycle after $2mD$
steps at most,
   where $D$ denotes the diameter of the graph.
 This result implies that the (edge) cover time of the rotor-router
model with a single token is $\Order(mD)$ for any graph.
 Bampas et al.~\cite{BGHI09} gave examples of which the stabilization
time gets ${\rm \Omega }(mD)$.
%
%
 Similar analyses for the rotor-router model with many tokens have
been developed, recently.
 Dereniowski et al.~\cite{DKPU14}
  investigated the cover time of the rotor-router model with $M$ tokens, and
  gave an upper $\Order (mD/\log M)$ and an example of ${\rm \Omega}(mD/M)$ as a lower bound.
 Chalopin et al.~\cite{CDGK15}
  gave an upper bound of its stabilization time is $\Order (m^4D^2+mD\log M)$,
  while they also showed that the period of a cyclic stabilized states
can get as large as $2^{\Omega(\sqrt{n})}$.

\begin{table}[t]
\begin{center}
\begin{tabular}{l | ll | ll}
\bhline{1.3pt}
Conditions on $P$
&$L_\infty$-discrepancy
&
&$L_1$-discrepancy
& 
\\ \bhline{1pt}
E. R.
& \multirow{2}{*}{$\Order\left(\frac{\Delta\log (n)}{1-\lambda^*}\right)$}
& \multirow{2}{*}{\cite{RSW98}}
& \multirow{2}{*}{$\Order\left(\frac{\Delta n\log (n)}{1-\lambda^*}\right)$}
& \multirow{2}{*}{}
\\
symmetric
&
&
&
&
\\ \hline
E.~R.~L.
& \multirow{2}{*}{$\Order(n|{\cal A}|)$}
& \multirow{2}{*}{\cite{KKM12}}
& \multirow{2}{*}{$\Order(n^2|{\cal A}|)$}
&
\\
rational
&
&
&
&
\\ \hline
any rational
& $\Order\left(\frac{\alpha^*n|{\cal A}|}{(1-\lambda^*)^\beta}\right)$
& \cite{KKM13}
& $\Order\left(\frac{\alpha^*n^2|{\cal A}|}{(1-\lambda^*)^\beta}\right)$
& 
\\ \hline
E.~R.
& $\Order\left( \frac{\pi_{\max}}{\pi_{\min}}t^* \Delta\right)$
& \cite{SYKY13}
& $\Order\left( \frac{\pi_{\max}}{\pi_{\min}}t^* \Delta n\right)$
& 
\\ \hline
E.~R.~L.
& \multirow{3}{*}{$\Order\left(d\min\left(\sqrt{\frac{\log(n)}{1-\lambda^*}},\sqrt{n}\right)\right)$}
& \multirow{3}{*}{\cite{BKKM15}}
& \multirow{3}{*}{$\Order\left(m\min\left(\sqrt{\frac{\log(n)}{1-\lambda^*}},\sqrt{n}\right)\right)$}
& 
\\
simple r.w.
&
&
&
&
\\
$d$-regular
&
&
&
&
\\ \bhline{1pt}
E.
&
&
& $\Order(mt^*)$
& Thm.~\ref{thm:upper-ergodic-Ob}
\\ \hline
E.~L.
&
&
& $\Order(m\sqrt{t^*\log t^*})$
& Thm.~\ref{thm:upper-lazyTV-SRT}
\\ \hline
E.~R.~L.
& \multirow{2}{*}{$\Order(\Delta\sqrt{t^*\log t^*})$}
& \multirow{2}{*}{Thm.~\ref{thm:upper-lazyLisrt}}
& \multirow{2}{*}{}
& 
\\
symmetric
&
&
&
&
\\
\bhline{1.3pt}
\multicolumn{5}{l}{E.: ergodic, R.: reversible, L.: lazy}
\end{tabular}
\caption{Summary of known results on
$\|\chi^{(t)}-\mu^{(t)}\|_{\infty}$ for finite graphs, and this work.
}
\label{table:upper}
\end{center}
\end{table}

\paragraph{Our results.}
 As we stated before,
  the total variation distance
   between the target distribution and approximate samples
  is significant in the analysis of MCMC algorithms.
 While there are several works on deterministic random walks
  concerning the vertex-wise discrepancy
$\|\chi^{(t)}-\mu^{(t)}\|_{\infty}$ such as~\cite{RSW98,KKM12,KKM13,SYKY13,BKKM15},
  little is known about the total variation discrepancy
$\|\chi^{(t)}-\mu^{(t)}\|_1$.
 This paper investigates the total variation discrepancy
  to develop a new analysis technique aiming at derandomizing MCMC.

 To begin with,
  we give a simple but nontrivial upper bound for any ergodic finite
Markov chains,
   precisely we show $\|\chi^{(t)}-\mu^{(t)}\|_1 = \Order(mt^*)$
   where $t^*$ is the mixing rate of $P$ and $m$ is the number of
edges of the transition diagram of~$P$.
 In fact,
   the analyses are almost the same for both the non-oblivious model,
including the rotor-router model~\cite{CS06,KKM12,KKM13,BKKM15},
   and the oblivious model like~\cite{RSW98, SYKY14} in which the states of routers are reset in
each step, and we in  
  Section~\ref{sec:result_oblivious} deal with the oblivious model.
 We also give a lower bound for the oblivious model presenting an
example such that $\|\chi^{(t)}-\mu^{(t)}\|_1=\Omega(t^*)$,
  which suggests that the mixing rate is negligible in the $L_1$
discrepancy for the oblivious model.

 Then, we in Section~\ref{sec:result_nonoblivious} give
   a better upper bound for non-oblivious determinstic random walk,
  precisely we show $\|\chi^{(t)}-\mu^{(t)}\|_1 =\Order(m\sqrt{t^*\log
t^*})$ when $P$ is ergodic and lazy.
 Notice that the upper bound does not require reversible.
 The analysis technique is a modification of Berenbrink et al.~\cite{BKKM15},
    in which they investigated a lazy version of simple random walks on $d$-regular graphs.
 In fact, we also remark that
   the analysis technique by~\cite{BKKM15} for the vertex-wise discrepancy is
   extended to general graphs, 
   precisely we show that $\|\chi^{(t)}-\mu^{(t)}\|_{\infty} =
\Order(\Delta \sqrt{t^*\log t^*})$
   when $P$ is ergodic, lazy, symmetric.
   We also present some lower bounds of $L_1$ discrepancy for non-oblivious models. 

 Table~\ref{table:upper} shows a summary of known
results~\cite{RSW98,KKM12,KKM13,SYKY13,BKKM15} on
$\|\chi^{(t)}-\mu^{(t)}\|_{\infty}$,
  and the results by this work.
 The column of ``$L_1$ discrepancy'' shows
   the upper bounds of $\|\chi^{(t)}-\mu^{(t)}\|_1$ implied by the
previous results~\cite{RSW98,KKM12,KKM13,SYKY13,BKKM15},
   in comparison with upper bounds obtained by this paper.

\section{Preliminaries}\label{sec:prel}
\subsection{Random walk / Markov chain}\label{subsec:RWMC}
As a preliminary step, we introduce some terminology of Markov chains (cf.~\cite{LPW08}).
%
%
Let $V=\{1, \ldots, n\}$ be a finite set, and let $P \in \mathbb{R}_{\geq 0}^{n \times n}$ be a transition matrix on $V$, which satisfies $\sum_{v\in V}P_{u,v}=1$ for any $v\in V$, where $P_{u,v}$ denotes the $(u,v)$ entry of $P$
 ($P^t_{u,v}$ denotes $(u,v)$ entry of $P^t$, as well).   
Let $\MDG=(V, \ME)$ be the transition digram of $P$, meaning that $\ME = \{(u,v) \in V\times V \mid P_{u,v}>0\}$. 
Let $\No(v)$ and $\Ni(v)$ respectively denote the out-neighborhood and the in-neighborhood of $v \in V$ on $\MDG$~\footnote{$\No(v) = \{ u \in V \mid P_{v,u}>0 \}$ and $\Ni(v) = \{ u \in V \mid P_{u,v}>0 \}$.}. 
For convenience, let $m=|\ME|$, $\deltao(v)=|\No(v)|$ and $\deltai(v)=|\Ni(v)|$. 
A finite Markov chain is called {\em ergodic} if $P$ is {\em irreducible}\footnote{$P$ is irreducible if $\forall u, v\in V, \exists t>0, P^t_{u, v} > 0$. Then, transition diagram of $P$ is connected.}
and {\em aperiodic}\footnote{$P$ is aperiodic if $\forall v\in V, \GCD\{ t \in \mathbb{Z}_{>0} \mid P^t_{v, v} > 0\} = 1$.}. 
It is well known that any ergodic $P$ has a unique {\em stationary distribution} $\pi \in \mathbb{R}_{\geq 0}^{n}$ (i.e., $\pi P = \pi$), and the limit distribution is $\pi$ (i.e., $\lim_{t\to \infty}\xi P^{t} = \pi$ for any probability distribution $\xi\in \mathbb{R}_{\geq 0}^{n}$ on $V$). 
Let $\xi$ and $\zeta$ be probability distributions on $V$, 
then the {\em total variation distance} $\dtv$ between $\xi$ and $\zeta$ is defined  by 
\begin{eqnarray}
\label{def:TV}
\dtv(\xi, \zeta)
\defeq \max_{A\subset V} \left| \sum_{v\in A}(\xi_v-\zeta_v )\right|
=\frac{1}{2} \left\|\xi-\zeta\right\|_1.
\label{def:dtv}
\end{eqnarray}
The {\em mixing time} of $P$ is defined by 
\begin{eqnarray}
 \tau(\varepsilon) \defeq 
 \max_{v \in V} \min \left\{ t \in \mathbb{Z}_{\geq 0} \mid \dtv(P^t_{v, \cdot}, \pi) \leq \varepsilon \right\}
\label{def:mix}
\end{eqnarray}
for any $\varepsilon > 0$~\footnote{$P^t_{v, \cdot}$ denotes the $v$-th row vector of $P^t$. }. 
Let $t^* \defeq \tau(1/4)$, called {\em mixing rate}, which is often used as a characterization of $P$. 
Let $\mu^{(0)}=(\mu^{(0)}_1,\ldots, \mu^{(0)}_n) \in \mathbb{Z}_{\geq 0}^{n}$ denote an initial configuration of $M$ tokens over $V$. 
Suppose that each token randomly and independently moves according to $P$.
Let $\mu^{(t)}$ denote the {\em expected} configuration of tokens at time $t\in \mathbb{Z}_{\geq 0}$ in a Markov chain, then $\mu^{(t)}=\mu^{(0)}P^t$ holds. 
By the definition of mixing time, $\|\mu^{(t)}/M-\pi\|_1\leq \varepsilon $ holds for any $t\geq \tau(\varepsilon )$ if $P$ is ergodic. 

\subsection{Deterministic random walk: framework}\label{subsec:detRW}
A {\em deterministic random walk} is a deterministic process imitating $\mu^{(t)}$.
Let $\chi^{(0)}=\mu^{(0)}$ and $\chi^{(t)}\in \mathbb{Z}_{\geq 0}^{n}$ denote the configuration of tokens at time $t\in \mathbb{Z}_{\geq 0}$ in a deterministic random walk. 
An update in a deterministic random walk is defined by $Z_{v,u}^{(t)}$ denoting the number of tokens moving from $v$ to $u$ at time $t$, where $Z_{v,u}^{(t)}$ must satisfy the condition that 
\begin{eqnarray}
 \sum_{u \in \No(v)} Z_{v,u}^{(t)}=\chi^{(t)}_v
 \label{eq:Zvut-usumsa}
\end{eqnarray}
for any $v\in V$. 
Then, $\chi^{(t+1)}$ is defined by
\begin{eqnarray}
 \chi_u^{(t+1)} \defeq \sum_{v \in \Ni(u)} Z_{v,u}^{(t)}
 \label{eq:Zvut-vsumsa}
\end{eqnarray}
for any $u\in V$. 
We will explain some specific deterministic random walks in Sections~\ref{subsec:model_oblivious} and \ref{subsec:model_srt} by giving precise definitions of $Z_{v,u}^{(t)}$.
We are interested in a question if $\chi^{(t)}$ approximates $\mu^{(t)}$ well in terms of the total variation discrepancy, i.e., the question is how large $\max_{A\subseteq V}|\chi^{(t)}_A-\mu^{(t)}_A\|=(1/2)\|\chi^{(t)}-\mu^{(t)}\|_1$ does get.

In the end of this section, we introduce two notations which we will use in the paper. 
For any $\xi \in \mathbb{R}^V$ and $A\subseteq V$, let $\xi_A$ denotes $\sum_{v\in A}\xi_v$. 
For example, $\mu^{(t)}_A=\sum_{v\in A}\mu^{(t)}_v$ and $P_{u,A}=\sum_{v\in A}P_{u,v}$. 
For any $\xi\in \mathbb{R}^n$, $P\in \mathbb{R}^{n \times n}$ and $u\in V$, let $(\xi P)_u$ denotes the $u$-th element of the vector $\xi P$, i.e., $(\xi P)_u=\sum_{v\in V}\xi_vP_{v,u}$.
%
\section{Upper and lower bounds for oblivious model}\label{sec:result_oblivious}
This section is concerned with an oblivious version of deterministic random walk, which is closely related to the models in~\cite{RSW98,SYKY14}. 
\subsection{Oblivious model}\label{subsec:model_oblivious}
%
Given a transition matrix $P$ and a configuration $\chi^{(t)}$ of tokens, we define $Z_{v,u}^{(t)}$ as follows. 
Assume that an arbitrary ordering $u_1,\ldots,u_{\deltao(v)}$ on $\No(v)$ is prescribed for each $v \in V$. Then, let 
\begin{eqnarray}
 Z_{v,u_i}^{(t)} = \left\{
 \begin{array}{ll}
  \left \lfloor \chi^{(t)}_v P_{v,u_i} \right \rfloor + 1 & (i \leq i^*) \\[2ex]
  \left \lfloor \chi^{(t)}_v P_{v,u_i} \right \rfloor     & (\mbox{otherwise})
 \end{array}\right.
\label{def:update}
\end{eqnarray}
where 
$i^* \defeq \chi^{(t)}_v - \sum_{i=1}^{\deltao(v)}  \lfloor \chi^{(t)}_v P_{v,u_i} \rfloor$ denotes the number of ``surplus'' tokens. 
It is easy to check that the condition \eqref{eq:Zvut-usumsa} holds for any $v,u\in V$ and $t\in \mathbb{Z}_{\geq 0}$. 
Then, the configuration $\chi^{(t+1)}$ is updated according to~\eqref{eq:Zvut-vsumsa}, recursively. 
The following observation is easy from the definition \eqref{def:update} of $Z_{v,u}^{(t)}$ .
\begin{observation}\label{obs:det-sampling}
For any oblivious model, 
$|Z_{v,u}^{(t)} - \chi^{(t)}_v P_{v,u}| \leq 1$
 holds for any $u,v\in V$ and $t\in \mathbb{Z}_{\geq 0}$. 
\end{observation}
\subsection{Upper bound}\label{subsec:upper_oblivious}
%
In this section, we give an upper bound of the total variation discrepancy.   
\begin{theorem}
\label{thm:upper-ergodic-Ob}
Suppose $P\in \mathbb{R}_{\geq 0}^{n\times n}$ is ergodic. 
Then, for any oblivious model, 
\begin{eqnarray*}
\left| \chi^{(T)}_A-\mu^{(T)}_A \right|
&\leq &\frac{3}{2}mt^*
=\Order(mt^*)
\end{eqnarray*}
holds for any $A \subseteq V$ and for any $T\in \mathbb{Z}_{\geq 0}$. 
\end{theorem}
%
Remark that Theorem~\ref{thm:upper-ergodic-Ob} only assumes that $P$ is ergodic. 
%
\begin{proof}[Proof of Theorem~\ref{thm:upper-ergodic-Ob}]
Let $\phi^{(t)}=\chi^{(t)}-\chi^{(t-1)}P$, for convenience.  
By \eqref{eq:Zvut-vsumsa} and Observation~\ref{obs:det-sampling}, 
\begin{eqnarray}
|\phi^{(t)}_u|
=\left |\left( \chi^{(t+1)}-\chi^{(t)}P\right)_u\right|
=\left| \sum_{v\in \Ni (u)} (Z_{v,u}^{(t)}-\chi_v^{(t)}P_{v,u})  \right|
\leq \sum_{v\in \Ni (u)} \left| Z_{v,u}^{(t)}-\chi_v^{(t)}P_{v,u}\right|
\leq \deltai(u)
\label{eq:phibound}
\end{eqnarray}
holds for any $u\in V$ and $t\in\mathbb{Z}_{\geq 0}$. 
Now, we see that
\begin{eqnarray}
\label{eq:maindisc1}
\sum_{t=0}^{T-1}\phi^{(T-t)}P^{t}
&=&\sum_{t=0}^{T-1}\left( \chi^{(T-t)} P^{t} - \chi^{(T-t-1)} P^{t+1} \right) 
=\chi^{(T)}P^0-\chi^{(0)}P^T
=\chi^{(T)} -\mu^{(T)}
\end{eqnarray}
hold, since $\mu^{(T)}=\chi^{(0)}P^T$ holds by the assumption. 
By \eqref{eq:maindisc1}, 
\begin{eqnarray}
\label{eq:discmix2}
\chi^{(T)}_A-\mu^{(T)}_A
&=&\left( \sum_{t=0}^{T-1}\phi^{(T-t)} P^{t} \right)_A 
=\sum_{t=0}^{T-1}\sum_{u\in V}\phi_u^{(T-t)} P^{t}_{u,A} \nonumber\\ 
&=&\sum_{t=0}^{\ats-1}\sum_{u\in V}\phi_u^{(T-t)} P^{t}_{u,A}
+\sum_{t=\ats}^{T-1}\sum_{u\in V}\phi_u^{(T-t)} \Bigl( P^{t}_{u,A}-\pi_A \Bigr)
\end{eqnarray} 
for any possible integer $\alpha$, where the last inequality follows from the fact that 
\begin{eqnarray*}
\sum_{u\in V}\phi^{(t)}_u=
\sum_{u\in V}\left( \chi^{(t+1)}-\chi^{(t)}P\right)_u 
&=&\sum_{u\in V}\chi^{(t+1)}_u-\sum_{u\in V}\sum_{v\in V}\chi^{(t)}_vP_{v, u}
=M-M
=0
\end{eqnarray*} 
holds for any $t\in \mathbb{Z}_{\geq 0}$. 
By \eqref{eq:discmix2}, we obtain that
\begin{eqnarray}
\Bigl|\chi^{(T)}_A-\mu^{(T)}_A \Bigr|
\leq \left|\sum_{t=0}^{\ats-1}\sum_{u\in V}\phi_u^{(T-t)} P^{t}_{u,A} \right|
+\left| \sum_{t=\ats}^{T-1}\sum_{u\in V}\phi_u^{(T-t)} \Bigl( P^{t}_{u,A}-\pi_A \Bigr) \right|. 
\label{eq:maindisc2}
\end{eqnarray}
Now, we give upper bounds of each term of \eqref{eq:maindisc2}. 
For the first term of \eqref{eq:maindisc2}, it is easy to see that
\begin{eqnarray}
\left|\sum_{t=0}^{\ats-1}\sum_{u\in V}\phi_u^{(T-t)} P^{t}_{u,A} \right|
\leq \sum_{t=0}^{\ats-1}|P^{t}_{u,A}| \sum_{u\in V}|\phi_u^{(T-t)}| 
\leq \sum_{t=0}^{\ats-1}\sum_{u\in V}\deltai(u)=m\alpha t^*
\label{eq:firstbound}
\end{eqnarray} 
holds by \eqref{eq:phibound}. 
To bound the second term of \eqref{eq:maindisc2}, we use the following lemma (See Appendix~\ref{app:markov} for the proof). 
\begin{lemma}
\label{lemm:dtsum}\cite{SYKY13}
Suppose $P \in \mathbb{R}_{\geq 0}^{n\times n}$ is {\em ergodic}. 
Then,  
\begin{eqnarray*}
 \sum_{t=\alpha t^*}^{\infty} \dtv\left( P^t_{u, \cdot}, \pi \right) 
 \leq \frac{t^*}{2^\alpha}
\end{eqnarray*}
holds for any $u\in V$ and for any $\alpha \in \mathbb{Z}_{>0}$. 
\shortqed
\end{lemma} 
By Lemma~\ref{lemm:dtsum}, we obtain that 
\begin{eqnarray}
\left| \sum_{t=\ats}^{T-1}\sum_{u\in V}\phi_u^{(T-t)} \Bigl( P^{t}_{u,A}-\pi_A \Bigr) \right|
\leq \sum_{t=\ats}^{T-1}\sum_{u\in V}|\phi_u^{(T-t)}| \Bigl| P^{t}_{u,A}-\pi_A \Bigr|
\leq \frac{t^*}{2^\alpha} \sum_{u\in V}\max_{0\leq t \leq T}|\phi_u^{(T-t)}| 
\leq \frac{mt^*}{2^\alpha}
\label{eq:secondbound}
\end{eqnarray}
hold where the last inequality follows from \eqref{eq:phibound}. 
Now, we obtain the claim from \eqref{eq:maindisc2}, \eqref{eq:firstbound} and \eqref{eq:secondbound} by letting $\alpha=1$. 
\end{proof}
\subsection{Lower bound}\label{subsec:lower_oblivious}
We give the following lower bound for an oblivious model. 
This proposition imply that we cannot improve the term $t^*$ for oblivious models in general. 
\begin{proposition}\label{prop:mixlower}
There exist an oblivious model such that
\begin{eqnarray*}
\max_{S\subseteq V}\left|\chi^{(T)}_S-\mu^{(T)}_S\right|={\rm \Omega}(nt^*)
\end{eqnarray*}
holds for any time $T$ after mixing. 
\end{proposition}
\begin{proof}
Let $V=\{0, \ldots, n-1\}$, and let a transition matrix $P$ be defined by $P_{u, u}=(k-1)/k$ for any $u\in V$, and $P_{u, v}=1/k(n-1)$ for any $u, v\in V$ such that $u\neq v$, 
i.e., $P$ denotes a simple random walk on $K_n$ with a self loop probability $(k-1)/k$ for any vertex. 
For this $P$, it is not difficult to check $t^*=\Order(k)$ (See Appendix~\ref{app:markov}). 
Then, we give a corresponding oblivious deterministic random walk. 
Let us assume that the prescribed ordering for each $v\in V$ starts with $v$ itself
(remember the definition of an oblivious deterministic random walk in Section~\ref{subsec:model_oblivious}). 
Let
\begin{eqnarray*}
\chi^{(0)}_u= \begin{cases}
    k & (u\in A)\\
    0 & (u\in B), 
  \end{cases}
\end{eqnarray*} 
where $A=\{0, \ldots, n/2-1\}$ and $B=\{n/2, \ldots, n-1\}$. 
Then, the initial configuration is stable, i.e., $\chi^{(t)}=\chi^{(0)}$, since each $v\in A$ serves $\lfloor k\cdot \frac{k-1}{k}\rfloor +1=k$ tokens to itself (notice that the ``surplus'' token stays at $v$ according to the prescribed ordering). 
Now it is easy to see that 
\begin{eqnarray*}
\max_{S\subseteq V}|\chi^{(t)}_S-\mu^{(t)}_S| &\geq &|\chi^{(t)}_A-\mu^{(t)}_A| 
\geq \frac{kn}{2}-\frac{kn}{4}-\varepsilon  
= \frac{kn}{4}-\varepsilon={\rm \Omega}(nt^*)
\end{eqnarray*}
holds for any $t\geq \tau(\varepsilon)$. We obtain the claim. 
\end{proof}
\section{Upper and lower bounds for non-oblivious model}\label{sec:result_nonoblivious}
%
Observation~\ref{obs:det-sampling} for oblivious model suggests only that $|Z_{v,u}^{(t)}-\chi^{(t)}_vP_{v,u}|\leq 1$ holds for any $t\in \mathbb{Z}_{\geq 0}$. 
In this section, we introduce the {\em SRT-router model} (c.f.,~\cite{SYKY13}), which satisfies $|\sum_{s=0}^{t}(Z_{v,u}^{(s)}-\chi^{(s)}_vP_{v,u})|\leq 1$ for any $t\in \mathbb{Z}_{\geq 0}$, 
and we obtain an improved bound when the Markov chain is {\em lazy}\footnote{$P$ is lazy if $P_{u,u}\geq 1/2$ holds for any $u\in V$.}.
\subsection{Model}\label{subsec:model_srt}
%
The SRT-router model, based on the {\em shortest remaining time} (SRT) rule~\cite{AJJ10,T80,SYKY13}, is a generalized version of the rotor-router model.  
In the model, we define an {\em SRT-router} $\sigma_v \colon \mathbb{Z}_{\geq 0} \to \No(v)$ on each $v \in V$ for a given $P$. 
Roughly speaking, $\sigma_v(i)$ denotes the destination of the $i$-th launched token at $v$.  
Given $\sigma_{v}(0), \ldots,\sigma_v(i-1)$, inductively $\sigma_v(i)$ is defined as follows. 
First, let 
\begin{eqnarray*} 
T_i(v)=\{u\in \No(v)\mid |\{ j\in[0,i)\mid \sigma_v(j) =u\}|-(i+1)P_{v, u}<0\}, 
\end{eqnarray*}
where $[z,z')\defeq \{z,z+1,\ldots, z'-1\}$ (remark $[z,z)=\emptyset$). 
Then, let $\sigma_v(i)$ be $u^*\in T_i(v)$ minimizing the value 
\begin{eqnarray*} 
\frac{|\{ j\in[0,i)\mid \sigma_v(j) =u\}|+1}{P_{v, u}}
\end{eqnarray*}
in any $u\in T_i(v)$. 
If there are two or more such $u\in T_v(i)$, then let $u^*$ be the minimum in them in an arbitrary prescribed order. 
The ordering $\sigma_v(0), \sigma_v(1), \ldots$ is known as the {\em shortest remaining time} (SRT) rule (see e.g.,~\cite{AJJ10,T80,SYKY13}). 
In an SRT-router model, there are $\chi^{(t)}_v$ tokens on a vertex $v$ at time $t$,   
and each vertex $v$ serves tokens on~$v$ to the neighboring vertices one by one according to $\sigma_v(i)$, like a rotor-router. 
For example, if there are $a$ tokens on $v$ at time $t=0$, then $|\{j\in [0,a)\mid \sigma_v(j)=u\}|$ tokens move to each $u\in \No(v)$, 
and there are $b$ tokens on $v$ at $t=1$, then $|\{j\in [a,a+b)\mid \sigma_v(j)=u\}|$ tokens move to each $u\in \No(v)$, and so on. 
Formally, it is defined by 
\begin{eqnarray}
\label{def:Zvufunc}
Z_{v,u}^{(t)}
&=&\left| \left\{ j\in \left[ \textstyle\sum_{s=0}^{t-1} \chi_v^{(s)}, \sum_{s=0}^{t} \chi_v^{(s)} \right) \mid \sigma_v(j)=u \right\} \right|. 
\end{eqnarray}
It is clear that the definition~\eqref{def:Zvufunc} satisfies \eqref{eq:Zvut-usumsa}. 
Then, the configuration of tokens is recursively defined by \eqref{eq:Zvut-vsumsa}. 
The following proposition is due to Angel et al.~\cite{AJJ10} and Tijdeman~\cite{T80}. 
\begin{proposition}
{\rm \cite{T80, AJJ10}}
\label{prop:fdiscSRT}
For any 
SRT-router model,  
\begin{eqnarray*}
\Bigl| |\{ j \in [0,z) \mid \sigma_v(j) =u\}| - z\cdotp P_{v,u}\Bigr| < 1
\end{eqnarray*}
 holds for any $v, u\in V$ and for any $z>0$. 
\end{proposition}
Proposition~\ref{prop:fdiscSRT} suggests that $|Z_{v,u}^{(t)}-\chi^{(t)}_vP_{v,u}|$ is small enough. In fact, 
Proposition~\ref{prop:fdiscSRT} and \eqref{def:Zvufunc} suggest a stronger fact that
\begin{eqnarray}
\left| \sum_{t=a}^{b}\left(Z_{v,u}^{(t)}-\chi_v^{(t)}P_{v,u}\right) \right|
&=& \left| \sum_{t=a}^b\left| \left\{ j\in \left[ \textstyle\sum_{s=0}^{t-1} \chi_v^{(s)}, \sum_{s=0}^{t} \chi_v^{(s)} \right) \mid \sigma_v(j)=u \right\} \right| - \sum_{t=a}^b\chi_v^{(t)}P_{v,u}  \right| \nonumber \\
&=& \left| \left| \left\{ j\in \left[ \textstyle \sum_{s=0}^{a-1} \chi_v^{(s)}, \sum_{s=0}^{b} \chi_v^{(s)} \right) \mid \sigma_v(j)=u \right\} \right| - \sum_{t=a}^b\chi_v^{(t)}P_{v,u}  \right| \nonumber \\
&\leq & \max_{\substack{z,z'\in \mathbb{Z}_{\geq 0}\\ {\rm s.t.}\ z'>z}}\Bigl| |\{ j \in [z,z') \mid \sigma_v(j) =u\}|-(z'-z)P_{v,u}\Bigr|
<2
\label{eq:discsrtsum}
\end{eqnarray}
holds for any $a,b\in \mathbb{Z}_{\geq 0}$ s.t. $a\leq b$. 
We will use \eqref{eq:discsrtsum} in our analysis, in Section~\ref{subsec:upper_srt}.  
%
\subsection{Better upper bound for the SRT-router model}\label{subsec:upper_srt}
Now, we show for ergodic and lazy $P$ the following theorem, modifying the technique~\cite{BKKM15}. 
\begin{theorem}
\label{thm:upper-lazyTV-SRT}
Suppose $P\in \mathbb{R}_{\geq 0}^{n\times n}$ is ergodic and lazy. 
Then for any SRT model,
\begin{eqnarray*}
\left| \chi^{(T)}_A-\mu^{(T)}_A \right|
=\Order\left(m\sqrt{t^*\log t^*}\right)
\end{eqnarray*}
holds for any $A \subseteq V$ and for any $T\in \mathbb{Z}_{\geq 0}$. 
\end{theorem}
%
%
\begin{proof}[Proof of Theorem~\ref{thm:upper-lazyTV-SRT}]
The major difference between an oblivious model and an SRT-router model is that 
\begin{eqnarray}
\left| \sum_{t=a}^{b}\phi^{(t)}_u \right|
= \left| \sum_{t=a}^{b}\sum_{v\in \Ni (u)} (Z_{v,u}^{(t)}-\chi_v^{(t)}P_{v,u}) \right| 
\leq \sum_{v\in \Ni (u)} \left| \sum_{t=a}^{b} (Z_{v,u}^{(t)}-\chi_v^{(t)}P_{v,u}) \right| 
\leq 2\deltai(u)
\label{eq:phibound2}
\end{eqnarray}
holds for any $u\in V$ and $b\geq a$ in an SRT-router model since \eqref{eq:discsrtsum} holds. 
It is easy to check that $|\chi^{(T)}_A-\mu^{(T)}_A |\leq 3mt^*$ holds for any SRT-router model by the same argument in the proof of Theorem~\ref{thm:upper-ergodic-Ob} using \eqref{eq:phibound2} instead of \eqref{eq:phibound}. 
Thus we obtain $|\chi^{(T)}_A-\mu^{(T)}_A |\leq 6m$ if $t^*=1,2$. 
In the rest part of the proof, we assume that $t^*\geq 3$, which suggests $t^*\lceil \lg t^*\rceil \geq 3$.
We introduce the following proposition and lemma to give a better upper bound of the first term of \eqref{eq:maindisc2}. See Appendix~\ref{app:markov} for the proofs.  
\begin{proposition}
\label{prop:summation}
Let $F_t=\sum_{i=0}^{t}f_i$. Then, 
\begin{eqnarray*}
\sum_{t=0}^{T}f_t g_t=F_T g_T +\sum_{t=0}^{T-1}F_t (g_{t}-g_{t+1})
\end{eqnarray*}
 holds for any $T\in \mathbb{Z}_{\geq 0}$ and for any $f_i, g_i$ $(0\leq i \leq T)$. 
 \shortqed
\end{proposition} 
\begin{lemma}
\label{lemm:lazysum}
Suppose that $P \in \mathbb{R}_{\geq 0}^{n\times n}$ is {\em ergodic} and {\em lazy}. 
Then, 
\begin{eqnarray*}
 \sum_{t=0}^{T} \dtv\left( P^t_{u, \cdot}, P^{t+1}_{u, \cdot} \right) 
 \leq 24\sqrt{T}-11
\end{eqnarray*}
holds for any $u\in V$ and for any $T\in \mathbb{Z}_{> 0}$. 
\shortqed
\end{lemma} 

Using Proposition~\ref{prop:summation}, \eqref{eq:phibound2} and Lemma~\ref{lemm:lazysum}, we obtain
\begin{eqnarray}
\left| \sum_{t=0}^{\ats -1}\phi^{(T-t)}_uP^t_{u,A} \right| 
&=& \left| \left( \sum_{i=0}^{\ats -1}\phi^{(T-i)}_u \right) P^{\ats-1}_{u,A}
+\sum_{t=0}^{\ats-2} \left( \sum_{i=0}^{t}\phi^{(T-i)}_u \right) \Bigl( P^t_{u,A}-P^{t+1}_{u,A} \Bigr) \right|\nonumber \\
&\leq & \left| \sum_{i=0}^{\ats -1}\phi^{(T-i)}_u \right| |P^{\ats-1}_{u,A}|
+\sum_{t=0}^{\ats-2} \left| \sum_{i=0}^{t}\phi^{(T-i)}_u \right|\Bigl| P^t_{u,A}-P^{t+1}_{u,A} \Bigr| \label{eq:rev} \\
&\leq & 2\deltai(u) +2\deltai(u)\cdot \Bigl( 24\sqrt{\ats-2} -11\Bigr)
=2\deltai(u) \Bigl( 24\sqrt{\ats-2} -10\Bigr)
\label{eq:firstbetter}
\end{eqnarray}
for any $u\in V$, where $\alpha$ is an arbitrary positive integer satisfying $\alpha t^*\geq 3$. 
Finally, \eqref{eq:maindisc2}, \eqref{eq:firstbetter}, \eqref{eq:secondbound} and \eqref{eq:phibound2} imply that
\begin{eqnarray*}
\left| \chi^{(T)}_A-\mu^{(T)}_A \right|
&\leq & \sum_{u\in V}\left|\sum_{t=0}^{\ats-1}\phi_u^{(T-t)} P^{t}_{u,A} \right|
+\frac{t^*}{2^\alpha} \sum_{u\in V}\max_{0\leq t \leq T}|\phi_u^{(T-t)}| \nonumber \\
&\leq &2m\Bigl( 24\sqrt{\ats-2} -10\Bigr) +2m\cdotp \frac{t^*}{2^\alpha}
\leq 2m\Bigl( 24\sqrt{t^*\lg t^*-2} -9\Bigr) 
\end{eqnarray*}
where the last inequality is obtained by letting $\alpha=\lceil \lg t^* \rceil$. We obtain the claim. 
\end{proof}
%
\subsection{Lower bounds}\label{subsec:lower_srt}
This section discusses a lower bound of the total variation discrepancy. 
First, we observe the following proposition, which is caused by the integral gap between $\chi^{(T)}\in \mathbb{Z}^V$ and $\mu^{(T)}\in \mathbb{R}^V$. 
\begin{proposition}
\label{prop1}
Suppose that $P$ is ergodic and its stationary distribution is uniform.
Then, for any $\chi^{(T)}\in \mathbb{Z}_{\geq 0}^n$ with an appropriate number of tokens $M$, 
\begin{eqnarray*}
\max_{S\subseteq V}\left|\chi^{(T)}_S-\mu^{(T)}_S\right|=\mathrm{\Omega}(n)
\end{eqnarray*}
holds for any time $T$ after mixing. 
\end{proposition}
%
%
We also give a better lower bound for an SRT-router model. 
\begin{proposition}
\label{prop2}
There exist an example of SRT model such that
\begin{eqnarray*}
\max_{S\subseteq V}\left|\chi^{(T)}_S-\mu^{(T)}_S\right|\geq
\frac{n^2}{8}=\mathrm{\Omega}(m)
\end{eqnarray*}
holds for any $T>0$.
\end{proposition}
See Appendix~\ref{app:markov} for the proofs. 

\subsection{Vertex-wise discrepancy}
This section presents an upper bound of the {\em single vertex discrepancy} $\|\chi^{(T)}-\mu^{(T)}\|_\infty$, 
which is an extended version of~\cite{BKKM15} to ergodic, reversible and lazy Markov chains, in general. 
\begin{theorem}
\label{thm:upper-lazyLisrt}
Suppose $P\in \mathbb{R}_{\geq 0}^{n\times n}$ is ergodic, reversible\footnote{ 
$P$ is reversible if the {\em detailed balance equation} $\pi_vP_{v,u}=\pi_uP_{u,v}$ holds for any $u,v\in V$. 
Notice that a reversible ergodic $P$ is symmetric if its stationary distribution is uniform, and vice versa.}, and lazy. 
Then for any SRT-router model, 
\begin{eqnarray*}
\left| \chi^{(T)}_w-\mu^{(T)}_w \right|
=\Order\left( \frac{\pi_{\max}}{\pi_{\min}} \Delta \sqrt{t^*\log t^*}\right)
\end{eqnarray*}
holds for any $w \in V$ and for any $T\in \mathbb{Z}_{\geq 0}$,
where $\Delta=\max_{u\in V}|\No(u)| (=\max_{u\in V}|\Ni(u)|)$, $\pi_{\max}=\max_{u\in V}\pi_u$ and $\pi_{\min}=\min_{u\in V}\pi_u$.
\end{theorem}
\begin{proof}
If $t^*=1,2$, $| \chi^{(T)}_w-\mu^{(T)}_w |\leq \frac{12\pi_{\max}}{\pi_{\min}} \Delta$ holds since $| \chi^{(T)}_w-\mu^{(T)}_w |\leq \frac{6\pi_{\max}}{\pi_{\min}} \Delta t^*$ holds due to~\cite{SYKY13}. 
Now, we assume $t^*\geq 3$, which suggests $t^*\lceil \lg t^* \rceil\geq 3$. By a combination of \eqref{eq:maindisc2}, \eqref{eq:rev}, \eqref{eq:secondbound} and \eqref{eq:phibound2}, we obtain that
\begin{eqnarray}
\Bigl|\chi^{(T)}_w-\mu^{(T)}_w\Bigr|
\leq 2\Delta \sum_{u\in V}|P^{\ats-1}_{u,w}|
+2\Delta\sum_{t=0}^{\ats-2} \sum_{u\in V}\Bigl|P^t_{u,w}-P^{t+1}_{u,w} \Bigr| 
+2\Delta \sum_{t=\ats}^{T-1} \sum_{u\in V}\Bigl|P^t_{u,w}-\pi_w\Bigr|
\label{eq:sngledisc}
\end{eqnarray}
holds, where $\alpha$ is an arbitrary positive integer satisfying $\alpha t^*\geq 3$. 
The condition that $P$ is reversible, i.e., $\pi_uP^t_{u,w}=\pi_{w}P^t_{w,u}$ holds for any $u,v\in V$, implies that
\begin{eqnarray}
\sum_{u\in V}P^t_{u,w}
=\sum_{u\in V}\frac{\pi_w}{\pi_{u}}P^t_{w,u}
\leq \frac{\pi_w}{\pi_{\min}}\sum_{u\in V}P^t_{w,u}
= \frac{\pi_w}{\pi_{\min}}
\label{eq:r1}
\end{eqnarray}
holds. Lemma~\ref{lemm:lazysum} implies that
\begin{eqnarray}
\sum_{t=0}^{\ats-2}\sum_{u\in V}\Bigl|P^t_{u,w}-P^{t+1}_{u,w} \Bigr|
&=&\sum_{t=0}^{\ats-2}\sum_{u\in V}\Bigl| \frac{\pi_w}{\pi_{u}}\Bigl( P^{t}_{w,u} -P^{t+1}_{w,u} \Bigr)\Bigr| 
\leq \frac{\pi_w}{\pi_{\min}}  \sum_{t=0}^{\ats-2}\sum_{u\in V}\Bigl| P^{t}_{w,u} -P^{t+1}_{w,u} \Bigr| \nonumber \\
&=& \frac{\pi_w}{\pi_{\min}}  \sum_{t=0}^{\ats-2} \| P^{t}_{w,u} -P^{t+1}_{w,u} \|_1
= \frac{2\pi_w}{\pi_{\min}} \sum_{t=0}^{\ats-2}\dtv \Bigl( P^t_{w,\cdot}, P^{t+1}_{w,\cdot}\Bigr) \nonumber \\
&\leq & \frac{2\pi_w}{\pi_{\min}} \Bigl(24\sqrt{\ats-2}-11 \Bigr)
\label{eq:r2}
\end{eqnarray}
holds, as well as Lemma~\ref{lemm:dtsum} implies that
\begin{eqnarray}
\sum_{t=\ats}^{T-1}\sum_{u\in V}\Bigl|P^t_{u,w}-\pi_w\Bigr|
&=&\sum_{t=\ats}^{T-1}\sum_{u\in V}\Bigl| \frac{\pi_w}{\pi_{u}}\Bigl( P^t_{w,u} -\pi_u \Bigr) \Bigr| 
\leq \frac{\pi_w}{\pi_{\min}} \sum_{t=\ats}^{T-1}\sum_{u\in V}\Bigl| P^t_{w,u} -\pi_u \Bigr| \nonumber \\
&=& \frac{2\pi_w}{\pi_{\min}} \sum_{t=\ats}^{T-1}\dtv \Bigl( P^t_{w,\cdot},\pi\Bigr) 
\leq \frac{2\pi_w}{\pi_{\min}} \frac{t^*}{2^\alpha}
\label{eq:r3}
\end{eqnarray}
holds. 
Thus, a combination \eqref{eq:sngledisc}, \eqref{eq:r1}, \eqref{eq:r2} and \eqref{eq:r3} implies that 
\begin{eqnarray*}
\Bigl|\chi^{(T)}_w-\mu^{(T)}_w\Bigr|
&\leq &2\Delta \frac{\pi_w}{\pi_{\min}}
+2\Delta \frac{2\pi_w}{\pi_{\min}} \Bigl(24\sqrt{\ats-2}-11 \Bigr)
+2\Delta \frac{2\pi_w}{\pi_{\min}} \frac{t^*}{2^\alpha} \\
&\leq & \frac{2\pi_w}{\pi_{\min}}\Delta \Bigl(48 \sqrt{t^*\lceil \lg t^* \rceil-2}-19 \Bigr)
\end{eqnarray*}
holds where the last inequality follows by letting $\alpha=\lceil \lg t^* \rceil$. We obtain the claim. 
\end{proof}
%
%
\section{Concluding Remarks}
In this paper, we gave two upper bounds of the {\em total variation discrepancy}, one is $\|\chi^{(t)}-\mu^{(t)}\|_1 =\Order(mt^*)$ for any ergodic Markov chains and
 the other is $\|\chi^{(t)}-\mu^{(t)}\|_1=\Order(m\sqrt{t^*\log t^*})$ for any lazy and ergodic Markov chains. 
We also showed some lower bounds. 
The gap between upper and lower bounds is a future work. 
Development of a deterministic approximation algorithm 
  based on deterministic random walks 
  for {\#}P-hard problems is a challenge. 

\bibliographystyle{abbrv}

\appendix
\section{Supplemental proofs}\label{app:markov} 
\subsection{Proof of Lemma~\ref{lemm:dtsum}}
For convenience, let $h(t)=\max_{u\in V}\dtv(P^t_{u,\cdotp},\pi)$.  
We use the following proposition to obtain Lemma~\ref{lemm:dtsum}. 
\begin{proposition}
\cite{SYKY13}
\label{prop:dltimes}
  For any integers $\ell$  $(\ell\geq 1)$ and 
  $k$ $(0 \leq k < t^*)$, 
\begin{eqnarray*}
h \left(\ell \cdotp t^*+k \right)\leq \frac{1}{2^{\ell +1}}
\end{eqnarray*}
 holds for any $u\in V$. 
\end{proposition} 
\begin{proof}[Proof of Lemma~\ref{lemm:dtsum}]
By Proposition~\ref{prop:dltimes}, 
\begin{eqnarray*}
\sum_{t=\ats}^{\infty} \dtv(P^t_{u,\cdotp},\pi)
&\leq &\sum_{t=\ats}^{\infty} h(t)
\leq \sum_{\ell =\alpha }^{\infty} \sum_{k=0}^{t^*-1}h(\ell t^*+k)
\leq \sum_{\ell =\alpha }^{\infty} \sum_{k=0}^{t^*-1}\frac{1}{2^{\ell+1}}
\leq t^*\cdotp \frac{1/2^{\alpha +1}}{1-1/2}
=\frac{t^*}{2^\alpha}
\end{eqnarray*}
holds. We obtain the claim. 
\end{proof}
\subsection{Supplemental proof of Proposition~\ref{prop:mixlower}}
We give a proof of $t^*=\Order(k)$ for Proposition~\ref{prop:mixlower}. 
\begin{proposition}
Let
\begin{eqnarray*}
P_{u,v}=\begin{cases}
    \frac{k-1}{k} & ({\rm if }\ v=u)\\
    \frac{1}{k(n-1)}  & ({\rm otherwise }). 
  \end{cases}
\end{eqnarray*}
Then 
\begin{eqnarray*}
\tau(\varepsilon)\leq \frac{n-1}{n-2}k\log \varepsilon^{-1}.
\end{eqnarray*}
\end{proposition}
\begin{proof}
The proof is based on the coupling technique~\cite{LPW08}. 
Let $X_t$ be a Markov chain according to $P$, and let $Y_t$ be another Markov chain with the same transition matrix $P$, where the transition from $Y_t$ to $Y_{t+1}$ depends on $X_t$ such that
\begin{eqnarray*}
Y_{t+1}= \begin{cases}
    Y_{t} & ({\rm if }\ X_{t+1}=X_{t})\\
    X_{t} & ({\rm if }\ X_{t+1}=Y_{t})\\ 
    X_{t+1} & ({\rm otherwise }). 
  \end{cases}
\end{eqnarray*} 
Then, it is not difficult to see that for any $X_0$ and $Y_0$, 
\begin{eqnarray*}
{\rm Pr}[X_t\neq Y_t]\leq \left(\frac{k-1}{k}+\frac{1}{k(n-1)}\right)^t
\end{eqnarray*} 
holds for any $t\in \mathbb{Z}_{\geq 0}$, thus $\dtv(P^t_{v, \cdot}, \pi)\leq \left(\frac{k-1}{k}+\frac{1}{k(n-1)}\right)^t$ by the coupling lemma (c.f.~\cite{LPW08}). Now, we obtain that
\begin{eqnarray*}
\tau(\varepsilon )&\leq &\frac{\log \varepsilon^{-1}}{\log \left(\frac{k-1}{k}+\frac{1}{k(n-1)}\right)^{-1}}
=\frac{\log \varepsilon^{-1}}{\log \left(1-\frac{n-2}{k(n-1)}\right)^{-1}} 
\leq \frac{\log \varepsilon^{-1}}{\frac{n-2}{k(n-1)}}
=\frac{n-1}{n-2}k\log \varepsilon^{-1}
\end{eqnarray*}
holds, where we used the fact that $\log (1-x)^{-1}\geq x$ holds for any $x\ (0<x<1)$. We obtain the claim. 
\end{proof}
\subsection{Proof of Proposition~\ref{prop:summation}}
\begin{proof}
Let $F_t=\sum_{i=0}^{t}f_i$. 
Then, $f_t=F_t-F_{t-1}$ holds. 
\begin{eqnarray*}
\sum_{t=0}^T f_t g_t
&=&f_0g_0+\sum_{t=1}^T f_t g_t
=f_0g_0+\sum_{t=1}^T (F_t-F_{t-1}) g_t \nonumber \\
&=&f_0g_0+\sum_{t=1}^T F_tg_t-\sum_{t=1}^T F_{t-1}g_t
=\sum_{t=0}^T F_tg_t-\sum_{t=0}^{T-1} F_{t}g_{t+1} \nonumber \\
&=&F_Tg_T+\sum_{t=0}^{T-1} F_{t}(g_t-g_{t+1}). 
\end{eqnarray*}
\end{proof}
\subsection{Proof of Lemma~\ref{lemm:lazysum}}
To bound $\dtv\left( P^t_{u, \cdot}, P^{t+1}_{u, \cdot} \right)$, we use the following proposition. 
\begin{proposition}
\cite{LPW08}
\label{lemm:sqrtsum}
Suppose $P\in \mathbb{R}_{\geq 0}^{n\times n}$ is ergodic and lazy.  
Then
\begin{eqnarray*}
\dtv (P^{t}_{u,\cdot},P^{t+1}_{u,\cdot})\leq \frac{12}{\sqrt{t}}
\end{eqnarray*}
holds for any $u\in V$ and for any $t>0$. 
\end{proposition}
\begin{proof}[Proof of Lemma~\ref{lemm:lazysum}]
By Proposition~\ref{lemm:sqrtsum}, 
\begin{eqnarray*}
\sum_{t=0}^{T} \dtv\left( P^t_{u, \cdot}, P^{t+1}_{u, \cdot} \right) 
&\leq &1+\sum_{t=1}^{T} \dtv \left( P^{t}_{u,\cdotp }, P^{t+1}_{u,\cdotp}  \right) 
\leq 1+\sum_{t=1}^{T}\frac{12}{\sqrt{t}} \nonumber \\
&\leq & 1+12\left( 2\sqrt{T}-1\right)
=24\sqrt{T}-11
\end{eqnarray*}
holds, and we obtain the claim. Remark that we use the fact $\sum_{t=1}^{T}\frac{1}{\sqrt{t}}\leq 2\sqrt{T}-1$. 
\end{proof}
\subsection{Proof of Proposition~\ref{prop1}}
\begin{proof}
Let $M=(k-1/2)n$ be the number of tokens for an arbitrary positive integer $k$. 
 Note that $\widetilde{\mu}^{(t)}_v = \mu^{(t)}_v/M$ converges to
$1/n$ for any $v \in V$
  since the stationary distribution is uniform. 
 Precisely, 
  for any $A \subseteq V$ and $T \geq \tau\left(1/(8k)\right)$, 
\begin{eqnarray}
\frac{|A|}{n} - \frac{1}{8k}
\leq \sum_{v \in A} \widetilde{\mu}^{(T)}_v
\leq \frac{|A|}{n} + \frac{1}{8k}
\end{eqnarray}
 holds
  by the definition (\ref{def:mix}) of the mixing time
$\tau(\varepsilon)$. 

Let $T$ be an arbitrary time, and let $A=\{ v \in V\mid \chi^{(T)}_v \geq k\} $. 
First, we consider the case that $|A| \geq n/2$. Then, we see that
$ \sum_{v \in A} \chi^{(T)}_v \geq k|A| $
holds.  
At the same time
\begin{eqnarray*}
 \sum_{v \in A} \mu^{(T)}_v
 &=& \sum_{v \in A} M \widetilde{\mu}^{(T)}_v
\leq \left(k-\frac{1}{2}\right) n \cdotp \left(\frac{|A|}{n} +
\frac{1}{8k}\right) 
\leq\left(k-\frac{1}{2}\right) |A| + \frac{n}{8}
\end{eqnarray*}
 holds. Thus
\begin{eqnarray*}
 \sum_{v \in A} \left(\chi^{(T)}_v - \mu^{(T)}_v \right)
&\geq& k|A| - \left(\left(k-\frac{1}{2}\right) |A| + \frac{n}{8}\right)
=\frac{1}{2} |A| - \frac{n}{8}
\geq \frac{n}{4}
\end{eqnarray*}
 where the last inequality follows $|A| \geq n/2$. 
 We obtain the claim in the case. 
Next, we consider the other case, meaning that $|A| < n/2$. 
Then, we see that
$ \sum_{v \in \overline{A}} \chi^{(T)}_v \leq (k-1) |\overline{A}| $
 since $\chi^{(T)}_v < k$ for any $v \in \overline{A}$. 
At that time,  
\begin{eqnarray*}
 \sum_{v \in \overline{A}} \mu^{(T)}_v
& = &\sum_{v \in \overline{A}} M \widetilde{\mu}^{(T)}_v
\geq \left(k-\frac{1}{2}\right) n \cdotp
\left(\frac{|\overline{A}|}{n} - \frac{1}{8k}\right) 
\geq \left(k-\frac{1}{2}\right) |\overline{A}| - \frac{n}{8}
\end{eqnarray*}
 holds. Thus
\begin{eqnarray*}
 \sum_{v \in \overline{A}} \left(\mu^{(T)}_v - \chi^{(T)}_v \right)
&\geq& \left(\left(k-\frac{1}{2}\right) |\overline{A}| +
\frac{n}{8}\right) - (k-1)|\overline{A}|
= \frac{1}{2} |\overline{A}| - \frac{n}{8}
\geq \frac{n}{4}
\end{eqnarray*}
 where the last inequality follows $|\overline{A}| \geq n/2$. 
 We obtain the claim. 
\end{proof}
\subsection{Proof of Proposition~\ref{prop2}}
\begin{proof}
We consider a random walk on a complete graph $K_{2n'}$, i.e., let $V=\{0, 1, \ldots, 2n'-1\}$ $(n'\in \mathbb{Z}_{>0})$ and $P_{u, v}=1/(2n')$ for any $u, v\in V$. 
Let $A=\{0, 1, \ldots, n'-1\}$, $B=\{n', n'+1\ldots , 2n'-1\}$ and let
\begin{eqnarray*}
\chi^{(0)}_u= \begin{cases}
    (2k+1)n' & (u\in A)\\
    0 & (u\in B), 
  \end{cases}
\end{eqnarray*} 
for an arbitrary $k\in \mathbb{Z}_{\geq 0}$. 
Note that $M=\|\chi^{(0)}\|_1=(2k+1)(n')^2$. 
Since this $P$ mixes in a single step, $\mu^{(t)}_A=\mu^{(t)}_B=(2k+1)(n')^2/2$ holds for any $t>0$. 
We define the SRT-router $\sigma_u(i)$ as
\begin{eqnarray*}
\sigma_u(i \bmod 2n')=i
\end{eqnarray*} 
for any $u\in V$. 
Then, it is not difficult to check that $\chi^{(t)}_A=(k+1)(n')^2$ and  $\chi^{(t)}_B=k(n')^2$ when $t$ is even, as well as that
$\chi^{(t)}_A=k(n')^2$ and $\chi^{(t)}_B=(k+1)(n')^2$ what is odd. 
Thus, 
\begin{eqnarray*}
\max_{S\subseteq V}|\chi^{(t)}_S-\mu^{(t)}_S| \geq |\chi^{(t)}_A-\mu^{(t)}_A|=\frac{(n')^2}{2}=\frac{n^2}{8}
\end{eqnarray*}
holds for any $t>0$. We obtain the claim. 
\end{proof}

\end{document}